\documentclass[11pt,reqno]{amsart}
\usepackage{amsmath,amssymb,amsthm,bbm,mathtools,calc,verbatim,enumitem,tikz,url,mathrsfs,cite,fullpage}
\usepackage{float}
\usepackage[hidelinks]{hyperref}
\usepackage{subcaption}

\usepackage{comment}
\usepackage{tabulary}
\usepackage{tabularx}
\usepackage{multirow}
\usepackage{booktabs}
\usepackage{cite}
\usepackage{enumitem}
\usepackage{float}
\usepackage{varioref}
\usepackage[capitalize]{cleveref}

\newtheorem{theorem}{Theorem}[section]

\newtheorem{corollary}[theorem]{Corollary}

\newtheorem{remark}[theorem]{Remark}
\crefname{cor}{Corollary}{Corollaries}
\theoremstyle{definition}

\theoremstyle{remark}
\newcommand{\prob}[5]{%
  \begingroup
  \par\medskip
  \noindent \textsc{#1}\nopagebreak[4]
  \par\noindent\hangindent=\parindent\textit{#2}  #3
  \par\noindent\hangindent=\parindent\textit{#4}  #5
  \par  \medskip
  \endgroup
}
\newcommand{\decprob}[3]{\prob{#1}{Input:}{#2}{Question:}{#3}}
\DeclareMathOperator{\vc}{vc}
\DeclareMathOperator{\tw}{tw}
\DeclareMathOperator{\fvs}{fvs}
\DeclareMathOperator{\td}{td}

\begin{document}

\title{On the Parameterized Complexity of Sparsest cut and small-set expansion problems}

\author{Ramin Javadi \and Amir Nikabadi}

\address{Department of Mathematical Sciences, Isfahan University of Technology, P.O. Box 84156-83111, Isfahan, Iran.
School of Mathematics, Institute for Research in Fundamental Sciences (IPM), P.O. Box: 19395-5746, Tehran, Iran.}\email{rjavadi@iut.ac.ir}

\address{Universit\'e Paris-Dauphine, Universit\'e PSL, CNRS, LAMSADE, 75016, Paris, France}\email{amir.nikabadi@dauphine.eu}

\begin{abstract}



 

We present a parameterized dichotomy for the \textsc{$k$-Sparsest Cut} problem in weighted and unweighted versions. In particular, we show that
the weighted \textsc{$k$-Sparsest Cut} problem
is NP-hard for every $k\geq 3$ even on graphs with bounded vertex cover number. Also, the unweighted \textsc{$k$-Sparsest Cut} problem is W[1]-hard when parameterized by the three combined parameters tree-depth, feedback vertex set number, and $k$. 
On the positive side, we show that
unweighted \textsc{$k$-Sparsest Cut} problem is FPT when parameterized by the vertex cover number and $k$ and when $k$ is fixed, it is FPT with respect to the treewidth. Moreover, we show that the generalized version \textsc{$k$-Small-Set Expansion} problem is FPT when parameterized by $k$ and the maximum degree of the graph, though it is W[1]-hard for each of these parameters separately.
\end{abstract}

\maketitle
\vspace{-1cm}
\section{Introduction}
All graphs in this paper are finite and simple.  We denote the set of positive integers by $\mathbb{N}$, and for every integer $k\in \mathbb{N}$, we denote by $[k]$ the set of all positive integers which are smaller than or equal to $k$. Let $G=(V,E)$ be a graph. For every $X$ which is either a vertex or a subset of vertices of $G$, we write $G\setminus X$ for the graph obtained from $G$ by removing $X$.
In this work, we study the problem of finding a $k$-partition of the vertices of a graph where each part has a low edge expansion. More precisely, let $G=(V,E)$ be a graph endowed with an edge weight function $w:E\to \mathbb{R}^+$ and let $S\subseteq V$ be a subset of vertices. \textit{The edge expansion} of $S$ is defined as $\phi_{G}(S) = \frac{w(S, \overline{S})}{|S|},$ where $w(S,\overline{S})$ stands for the sum of the weights of all edges in $G$ with exactly one endpoint in $S$.
We drop the subscript $G$ when there is no ambiguity. The \textsc{Sparsest Cut Problem}  asks for a subset $S \subseteq V$ with at most $|V|/2$ vertices which has the least edge expansion. Let
\[\phi(G) = \min_{\stackrel{S\subseteq V}{|S|\leq |V|/2}} \phi_G(S)=  \min_{\stackrel{S\subseteq V}{S\neq \emptyset, S\neq V}} \frac{w(S, \overline{S})}{\min\{|S|,|\overline{S}|\}}. \] 
\decprob{\textsc{Sparsest Cut}}{A graph $G=(V,E)$, a weight function $w:E\to \mathbb{R}^+$ and a rational number $N$.}{Does there exist a subset $S\subseteq V$, where $|S|\leq |V|/2$ and $\phi_G(S)\leq N$? }
The \textit{unweighted} version of the problem is when all edge weights are equal to one. A natural generalization of \textsc{Sparsest Cut} problem is to find a $k$-partition of $V$ such that the worst edge expansion of the parts is minimized. More precisely, let $k\geq 2$ be an integer and $\mathcal{S}=\{S_1,\ldots, S_k \}$ be a partition of $V$ into $k$ subsets. Define,
\[ \phi_k(G)=\min_{\mathcal{S}=\{S_1,\ldots, S_k\}} \max_{1\leq i\leq k} \phi_{G}(S_{i}), \] where the minimum is taken over all $k$-partitions of $V$. This generalization is called \textsc{$k$-Sparsest Cut}. One may see that \textsc{Sparsest Cut} problem is the special case when $k=2$. 

\decprob{\textsc{$k$-Sparsest Cut} ($k$SC)}{A graph $G=(V,E)$, a nonnegative integer $k$ and a rational number $N$.}{Does there exist a $k$-partition $\{S_1,\ldots, S_k\}$ of $V$ where the edge expansion of each part is at most $N$, i.e. for every $i\in [k]$, $\phi_G(S_i)\leq N$?}%

We refer to another generalization of \textsc{Sparsest Cut} problem when we restrict the search space into small subsets of $V(G)$. Let $k$ be a positive integer. The \textsc{Small-Set Expansion} problem seeks for a subset $S\subseteq V$ of size at most $k$ with the minimum edge expansion. We define,
\[\psi_k(G)=\min_{\stackrel{S\subseteq V(G)}{|S|\leq k}} \phi_G(S). \]
\decprob{\textsc{$k$-Small-Set Expansion} ($k$SSE)}{A graph $G=(V,E)$, a positive integer $k$ and a rational number $N$.}{Does there exist a subset $S\subseteq V$ of size at most $k$ such that $\phi_G(S)\leq N$?}%

If $\{S_1,\ldots,S_k \}$ is a $k$-partition of $V$, then there is some $i\in [k]$, where $|S_i|\leq |V|/k$. Therefore, for every integer $k\geq 2$, we have 
\begin{equation} \label{eq:phipsi}
 \psi_{|V|/k}(G)\leq \phi_k(G). 
\end{equation}
It is readily observed that:
\begin{equation} \label{eq:phi2}
\phi_2(G)= \min_{S\subsetneq V(G)} \max(\phi_G(S),\phi_G(\overline{S})\}
= \min_{\stackrel{S\subsetneq V(G)}{|S|\leq |V|/2}} \phi_G(S)= \psi_{|V|/2}(G).
\end{equation}
Thus, equality holds in \eqref{eq:phipsi} when $k=2$. We refer the reader to the two classic results regarding the \textsc{Sparsest Cut} problem, Leighton and Rao's $O(\log n)$ approximation algorithm \cite{leighton1989approximate}, and Arora, Rao, and Vazirani's $O(\sqrt{\log n})$ approximation algorithm \cite{arora2009expander}. Daneshgar et al. \cite{daneshgar2012complexity,daneshgar2013clustering} showed that $k$\textsc{SC} is NP-hard even for trees and gave an $O(n \log n)$ time algorithm for weighted trees when the search space is relaxed to all $k$-subpartitions ($k$ disjoint subsets of vertices). 
 
\subsection*{Our contribution:}
We study the parameterized complexity of $k$-\textsc{Sparsest Cut} problem and $k$-\textsc{Small-Set Expansion} problem where we focus on graphs of bounded treewidth, bounded vertex cover number, bounded degree, and degeneracy. We divide the investigation into weighted and unweighted graphs. The problem $k$SC shows different complexity behavior in weighted and unweighted versions. \Cref{tbl:1} gives an overview of our obtained results. Also, \Cref{fig:results} shows the relationship between the parameters.
We begin by presenting our results for weighted versions of $k$-SC and $k$SSE in \Cref{sec3}. We prove that $k$SSE and $2$SC are FPT with respect to the treewidth and $k$SC, for every fixed $k\geq 3$, is NP-hard even on graphs with bounded vertex cover number $\vc$. We also prove the NP-hardness of $k$SC and $k$SSE for every fixed $k\geq 2$, even for graphs of bounded maximum degree and bounded degeneracy. In \Cref{sec4}, we investigate the unweighted version of $k$SC and we prove that for every fixed $k\geq 2$, the problem $k$SC, for every fixed $k$, is FPT with respect to the treewidth $\tw$. Although in the running time of the algorithm, $k$ is in the exponent, we prove that it is unlikely to improve it to $f(k,\tw)\, n^{O(1)}$ by showing that unweighted $k$SC is W[1]-hard w.r.t. the parameters $k$ and treewidth, combined (and even more, w.r.t. $k$, feedback vertex set number and tree-depth). Finally, we prove that unweighted $k$SC happens to be FPT w.r.t. the parameters $k$ and the vertex cover number. \Cref{sec5} begins with proving W[1]-hardness of $k$SSE for the parameter $k$. 
Also, based on the random separation method of Cai et al. \cite{cai2006random}, we provide an FPT algorithm for $k$SSE w.r.t. $k$ and the maximum degree of the input graph, combined.
\vspace{-0.1cm}

\begin{table}[H]
	\centering\footnotesize
 \caption{\label{tbl:1}Overview of known and new parameterized results for weighted$\backslash$unweighted $k$SC and $k$SSE. Herein, $\Delta, d, \vc$, tw, fvs, td denotes the maximum degree, the degeneracy, the vertex cover number, the treewidth, the feedback vertex set number, and the tree-depth of the input graph, respectively.}
	\begin{tabular}{p{.09\textwidth} p{.1\textwidth} p{.44\textwidth}}     
		\toprule
		Problem & Parameter & Results \\%
	
		\midrule
		Weighted $k$SC
		& $\vc$ &
		FPT for $k=2$ (\Cref{thm:vck=2})\newline
		NP-hard for every $k\geq 3$ and $\vc\geq 3$ (\Cref{thm:NPtau3}) \\
		\cmidrule{2-3}
		& tw &
		FPT for $k=2$ (\Cref{thm:twk=2})\newline
		NP-hard for every $k\geq 3$ and $\tw\geq 3$ (\Cref{cor:NPtw3})\newline
		Polynomial for every $k\geq 2$ and $\tw=1$ (\hspace{-0.01mm}\cite{daneshgar2012complexity}) \\
		\cmidrule{2-3}
		& $\Delta$&
		NP-hard for every $k\geq 2$ and $\Delta \geq 3$ (\Cref{thm:2})\\
		\cmidrule{2-3}
		& $d$ &
		NP-hard for every $k\geq 2$ and $d \geq 2$  (\Cref{thm:2})\\
		\midrule
		{Unweighted $k$SC}&
		($k, \fvs, \td$)&
		W[1]-hard (\Cref{thm:w1})\\
		\cmidrule{2-3}
		&($k,\tw$) &
		W[1]-hard (\Cref{cor:tw})\\
  	\cmidrule{2-3}
		&$(k,\vc)$ &
		FPT (\Cref{thm:vc})\\
		\cmidrule{2-3}
		& tw &
		FPT for fixed $k\geq 2$ (\Cref{thm:tw})\\
	
		\midrule
		{$k$SSE}&
		$k$ &
		W[1]-hard even for unweighted  (\Cref{thm:kssew1k})\\
		\cmidrule{2-3}
		&$\Delta$ &
		NP-hard for every $\Delta \geq 3$ (\Cref{cor:delta})\\
		\cmidrule{2-3}
		
		&($k, \Delta$)&
		FPT (\Cref{thm:randomsep})\\
		\cmidrule{2-3}
		& tw &
		FPT  (\Cref{thm:twk=2})\\
		\cmidrule{2-3}
		&$\vc$ &
		FPT  (\Cref{thm:vck=2})\\
		\bottomrule
	\end{tabular}
	\label{table:results}
\end{table}

\vspace{-1cm}

\begin{figure}[htb]
 \centering

 \begin{tikzpicture}[thick,scale=1.5, every node/.style={scale=1.1}]
 \node[] (a) at (0,0) {$\vc+k$};
 \node[] (b) at (0,1) {$\td+k$};
 \node[] (c) at (0,2) {$\text{pw}+k$};
 \node[] (d) at (0,3) {$\tw+k$};
 \node[] (e) at (-2,1.7) {$\fvs+k$};
 
   %
 \path[->, very thick] (a) edge node {} (b);
 \path[-> , very thick] (b) edge node {} (c);
 \path[-> , very thick] (c) edge node {} (d);
 \path[-> , very thick] (a) edge node {} (e);
 \path[-> , very thick] (e) edge node {} (d);
 
%
  \node[] (c1) at (-2,0.5) {};
   \node[] (c2) at (2,0.5) {};
\path[very thick, dashed, red] (c1) edge [in=180, out=30] (c2);

 \node[] at (3.1,0.5) {\textsc{Unweighted $k$SC}};

\node[] (b1) at (-1.5,3.2) {};
\node[] (b2) at (1.5,3.2) {};
\path[very thick, dashed, blue] (b1) edge [in=170, out=0] (b2);

 \node[] at (0,3.8) {\textsc{Weighted 2SC, $k$SSE,}};
  \node[] at (0,3.5) {\textsc{Unweighted $k$SC} (fixed $k\geq 2$)};

\node[] (c1) at (-2,0) {};
\node[] (c2) at (2,0) {};
\path[very thick, dashed, red] (c1) edge [in=-150, out=-30] (c2);

 \node[] at (3.7,0) {\textsc{Weighted $k$SC} (fixed $k\geq 3$)};
 
 \end{tikzpicture}
  \caption{\label{fig:results}Relationship and summary of the results for considered 
 parameters. Arrows indicate more general parameters, that is, an arrow $p\rightarrow q$ between two parameters $p,q$ shows that there exists a computable function $g$ such that $q \leq g(p)$. Red dotted lines show the tractability borders (FPT) for the parameter.}
 \end{figure}


\section{Preliminaries}\label{sec:prelim}
We use standard graph-theoretic notation, see\cite{bondy1976graph}. We denote the open neighborhood of a vertex $v$ in $G$ by $N_{G}(v)$. 
The size of $N_G(v)$ is called the degree of $v$ and the maximum degree of all vertices is denoted by $\Delta$. Given a subset $S \subseteq V(G)$, the subgraph of $G$ induced by $S$ is denoted by $G[S]$. A graph $G$ is called $d$-\textit{degenerate} if every induced subgraph of $G$ has a vertex of degree at most $d$. The minimum number $d$ for which $G$ is $d$-generate is called the \textit{degeneracy} of $G$. We assume the reader is familiar with basic notions in parameterized complexity, such as the FPT and W[1] classes, see  \cite{cygan2015parameterized}. We rely on a number of well-known graph structural parameters: treewidth \cite{cygan2015parameterized}, tree-depth \cite{nevsetvril2006tree}, feedback vertex set, and vertex cover denoted respectively as $\tw(G)$, $\td(G)$, $\fvs(G)$, $\vc(G)$, where we drop $G$ if it is clear from the context. The Exponential Time Hypothesis (ETH)   \cite{impagliazzo2001problems} states that 3-SAT
 on instances with $n$ variables and $m$ clauses cannot be solved in time $2^{o(n+m)}$.

\section{Weighted version}\label{sec3}
In this section, we show that $k$SSE and $2$SC becomes FPT when parameterized by the treewidth (\Cref{thm:twk=2}), via a treewidth-based dynamic programming algorithm. Moreover, in \Cref{cor:NPtw3} we will show that this result cannot be extended to $k$SC for $k\geq 3$.

\begin{theorem}\label{thm:twk=2}
The problem $k$SSE, for every $k$, and $2$SC parameterized by the treewidth are fixed-parameter tractable. Also, if the input graph $G$ has $n$ vertices and its tree decomposition of width $\tw$ is given, then the algorithm runs in $O(2^{\tw}k^2n)$ time for $k$SSE and in  $O(2^{\tw}n^3)$ for $2$SC and uses exponential space to $\tw$.
\end{theorem}

\begin{proof}
First, note that due to Equation \eqref{eq:phi2}, $2$SC is a special case of $k$SSE for $k=n/2$. So, we only prove it for $k$SSE. The proof is based on dynamic programming which computes the values of a table on the nodes of the tree decomposition of the graph in a bottom-up fashion. For convenience and easier analysis, we use a nice tree decomposition.

Suppose that $\tw(G)=\tw$ and consider a nice tree decomposition $ (T, \{ X_{t} \}_{t\in V(T )} ) $ for $G$ of width $\tw$ as defined in Section~\ref{sec:prelim}. For each node $i\in V(T)$, let $T_{i}$ be the subtree of $T$ rooted at $i$ and $G_{i} = (V_{i}, E_{i})$ be the subgraph of $G$ induced by the vertices in $\bigcup_{j \in V(T_{i})} X_{j}$. 
For each node $i\in V(T)$, every integer $1\leq s\leq k$ and every subset $C\subseteq X_i$, define 
\[A_i[C,s]=\min_{\stackrel{S\subset V_i}{\stackrel{|S|=s}{S\cap X_i=C}}} \phi_{G_i}(S). \]
The value of $A_i[C,s]$ is defined to  be $+ \infty$ when there is no feasible solution.
Now, we comprise a table where each row represents a node of $T$ (from leaves to the root), and each column represents an integer $s$, $1\leq s\leq k$ and a subset $C\subseteq X_i$. The value of row $i$ and column $(C,s)$ is equal to $A_i[C,s]$. The algorithm examines the nodes of $T$ in a bottom-up manner and fills in the table by the following recursions.

\subsection*{Leaf nodes.}
In the initialization step, for each leaf $i\in V(T)$, we have $V_i=X_i=\emptyset$. Therefore,

\[A_i[\emptyset,s]= +\infty, \quad \forall \  1\leq s\leq k
\]
\subsection*{Forget nodes.}
Let $i$ be a forget node with a child $j$, where $X_i=X_j\setminus \{v\}$ and $v\in X_j$. Then, the vertex $v$ is either inside or outside the solution $S$. Therefore,
\[A_i[C,s]=\min\{A_j[C,s], A_j[C\cup\{v\},s] \}. \]
\subsection*{Introduce nodes.}
Let $i$ be an introduce node with a child $j$, where $X_i=X_j\cup\{v\}$ and $v\notin X_j$. For integer $1\leq s\leq k$ and subset $C\subseteq X_i$, we have
\[A_i[C,s]=\begin{cases}
A_j[C,s]+\dfrac{w(\{v\},C)}{s} & \text{if }v\notin C, \\
\dfrac{s-1}{s} A_j[s-1,C\setminus \{v\}]+\dfrac{w(\{v\},X_i\setminus C)}{s} & \text{if }v\in C. 
\end{cases} \]

\subsection*{Join nodes.}
Let $i$ be a join node with children $j$ and $k$, where $X_i=X_j=X_k$. Then,
\[A_i[C,s]= \min_{s_1+s_2=s+|C|} \frac{s_1A_j[C,s_1]+s_2A_j[C,s_2]-w(C,X_i\setminus C)}{s}.
\]
Finally, the solution to $k$SSE for $G$ is equal to $\min_{C\subseteq X_r, 1\leq s\leq k} A_r[C,s]$, where $r$ is the root of $T$. Since the size of the table is at most $n\times k2^{\tw}$ and the value for each join node is computed in $O(k)$, the runtime of the whole algorithm is at most $2^{\tw}O(nk^2)$. Also, for a graph $G$ of treewidth $\tw$, a tree decomposition of width $\tw$ can be found in time $\tw^{O(\tw^3)}O(n)$ \cite{bodlaender1996linear}. Hence, the problem $k$SSE parameterized by the treewidth is FPT.
\end{proof}

Since the treewidth of the graph is bounded by its vertex cover number, \Cref{thm:twk=2} implies that $k$SSE and $2$SC with the parameter vertex cover are both in FPT. However, the algorithm uses an exponential space. In the following theorem, when the size of the vertex cover is bounded,  we give an alternative algorithm whose runtime is better and uses polynomial space. 

\begin{theorem} \label{thm:vck=2}
The problems $k$SSE, for every $k$, and $2$SC can be solved in respectively $O(2^{\vc}kn \log n)$ and $O(2^{\vc}n^2 \log n)$ time and uses polynomial space, where $n$ and $\vc$ are respectively the order and the vertex cover number of the input graph.
\end{theorem}
\begin{proof}
As in the proof of Theorem~\ref{thm:twk=2}, we only prove it for $k$SSE. Let $C$ be a vertex cover of the graph $G$ with size $|C|=\vc$ and define $I=V(G)\setminus C$. Fix a subset $\tilde{C}\subseteq C$ and an integer $1\leq s\leq k$ and define 
\begin{equation} \label{eq:11a}
\phi[\tilde{C},s]= \min_{\stackrel{S\subsetneq V(G)}{|S|=s,\ S\cap C=\tilde{C}}} \phi(S).
\end{equation}

When there is no feasible solution, define $\phi[\tilde{C},s]$ to be $+\infty$. 

We show that for fixed $\tilde{C},s$, the value of $\phi[\tilde{C},s]$ can be found in the time polynomial in $n$. For every vertex $i\in I$, define $w_i=w(\{i\},\tilde{C})$ and $w'_i=w(\{i\},C\setminus \tilde{C})$. Also, define $W=w(\tilde{C},\overline{\tilde{C}})$. Now, let $S\subseteq V(G)$ be such that $|S|=s$ and $S\cap C=\tilde{C}$. Then,
\begin{align} 
\phi(S)&= \frac{\sum_{i\in S\cap I} w'_i+ \sum_{i\in I\setminus S}w_i+w(\tilde{C},C\setminus \tilde{C})}{s}\notag \\
&=
 \frac{\sum_{i\in S\cap I} w'_i+W-\sum_{i\in S\cap I} w_i}{s}=\frac{W}{s}+\frac{1}{s} \sum_{i\in S\cap I} (w'_i-w_i). \label{eq:11b}
\end{align}

Now, sort the vertices in $I$ by the values $w''_i=w'_i-w_i$, $i\in I$, in nondecreasing order. Without loss of generality, suppose that $I=\{1,2,\ldots ,n-\vc\}$ and $w''_1\leq w''_2\leq \cdots w''_{n-\vc}$. If $s<|\tilde{C}|$, then clearly, $\phi[\tilde{C},s]=+\infty$. So, suppose that $s\geq |\tilde{C}|$ and define $s'=s-|\tilde{C}|$ and $\tilde{S}=\tilde{C}\cup \{1,2,\ldots, s'\}$. Hence, by \eqref{eq:11a} and \eqref{eq:11b}, we have 
$\phi[\tilde{C},s]= \phi(\tilde{S})$.

Finally, we have 
\[\psi_k(G)= \min_{\stackrel{\tilde{C}\subseteq C}{1\leq s\leq k}} \phi[\tilde{C},s]. \]

The computation of $\phi[\tilde{C},s]$ takes $O(n\log n)$ time and polynomial space using a simple sorting algorithm. Thus, the algorithm runs in $O(2^{\vc} kn\log n)$ time and uses polynomial space. 
\end{proof}

One may naturally ask if \Cref{thm:twk=2,thm:vck=2} can be generalized for the problem $k$SC when $k\geq 3$. In the following theorem, we show that the answer is negative. 
\begin{theorem}\label{thm:NPtau3}
For every fixed integer $k\geq 3$ and $\vc\geq 3$, the problem $k$SC is NP-hard for graphs with vertex cover number at most $\vc$.
\end{theorem}
\begin{proof}
Let $k\geq 3$ be a fixed integer. We are going to prove that $k$SC is NP-hard for graphs with vertex cover  numbers at most $3$. We give a polynomial reduction from \textsc{Partition} problem which is well-known to be NP-hard \cite{garey}.

\decprob{\textsc{Partition}}{Positive integers $w_{1}, . . . , w_{n}, B$, where $\sum_i w_i=2B$.}{Does there exist a subset $I\subseteq \{1,\ldots,n \}$ such that $\sum_{i\in I} w_i=B$?}

Let $w_1,\ldots, w_n$ be an instance of \textsc{Partition}. Let us define $w_{n+1}=B$.  We construct a graph $G$ with a vertex cover number equal to three and a number $N$ such that the answer to \textsc{Partition} is yes if and only if $\phi_k(G)\leq N$. Let $M$ be a fixed integer that will be determined shortly and define $G$ to be the bipartite graph with bipartition $(X,Y)$, where $X=\{u_1,u_2,u_3\}$ and 
$Y=V\cup \hat{Y}\cup U_1\cup U_2\cup U_3$, where 
$ V=\{v_1,\ldots, v_{n+1} \} $, $\hat{Y}=\{y_1,\ldots, y_{k-3}\}$ and $U_t=\{u^t_1,\ldots,u^t_{M-1} \}$, $t\in\{1,2,3\}$.
Also, all vertices in $V\cup \hat{Y}$ are adjacent to every vertex in $X$ and every vertex in $U_t$ is adjacent to $u_t$, $t=1,2,3$ (see Figure \ref{fig:cbg}). Moreover, define $N=4B/ (M-n-1-(k-3)/3)$ and choose $M$ such that $M> n+1+(k-3)/3$. The edge weights are defined as follows.

\begin{align*}
w(v_iu_t)&= w_i+N,\ 1\leq i\leq n+1,\ 1\leq t\leq 3 \\
w(y_{j}u_t) &= N/3,\ 1\leq j\leq k-3,\ 1\leq t\leq 3\\
w(u_l^tu_t) &= +\infty,\ 1\leq l\leq M-1,\ 1\leq t\leq 3. 
\end{align*}

Let $I\subset \{1,\ldots, n\}$ be such that $w(I)=\sum_{i\in I}w_i =B$. Also, let $V_I=\{v_i, i\in I \}$, $V_{\overline{I}}=\{v_i, i\not\in I \}$. 
Define, 
\begin{align*}
S_1&=\{u_1,u^1_l, 1\leq l\leq M-1\}\cup V_I,\\
S_2&=\{u_2,u^2_l, 1\leq l\leq M-1\}\cup V_{\overline{I}},\\
S_3&=\{u_3,u^3_l, 1\leq l\leq M-1\}\cup \{v_{n+1}\},\\
S_j&=\{y_{j-3}\}, \ 4\leq j\leq k.
\end{align*}
We leave to the reader to see $\phi(S_i)= N$, for every $i\in \{1,\ldots, k\}$. As an example of this, below we compute the $\phi(S_1)$.

\begin{align*}
\phi(S_1)&=
\frac{w(V_I,\{u_2,u_3\})+w(\{u_1\},\hat{Y})+w(\{u_1\},V_{\overline{I}})}{|I|+M}\\
&= \frac{2(w(I)+N|I|)+(k-3)N/3+3B+(n+1)N - w(I)-N|I|}{|I|+M}\\
&= \frac{w(I)+N|I|+(k-3)N/3 +3B+(n+1)N}{|I|+M}\\
&=\frac{4B+N(n+1+(k-3)/3)+N|I|}{|I|+M}= \frac{NM+N|I|}{|I|+M}=N.
\end{align*} 
For the converse, suppose that $\{S_1,\ldots, S_k\}$ be a $k$-partition of $V(G)$ such that $\phi(S_i)\leq N $, for each $i\in\{1,\ldots, k\}$. For $t\in \{1,2,3\}$, define $U_t'=U_t \cup\{u_t\}$. Since the weight of the edge $u_tu^t_l$ is equal to $+\infty$, each $U'_t$ is completely included in one subset $S_i$.

If for some $i$, $S_i\cap \{v_1,\ldots, v_{n+1}\}$ is non-empty, then $S_i \cap (U'_1\cup U'_2\cup U'_3)$ is also non-empty. For if $S_i\cap (U'_1\cup U'_2\cup U'_3)=\emptyset $, then $\phi(S_i)> \frac{N|S_i|}{|S_i|}=N$ which is a contradiction. 

Now, since $|\hat{Y}|=k-3$, there exist at least three subsets, say $S_1,S_2,S_3$, such that for each $i\in\{1,2,3\}$, $S_i\cap \hat{Y}$ is empty. By the above argument, for each $i\in \{1,2,3\}$, $S_i$ has non-empty intersection with $U'_1\cup U'_2\cup U'_3$. Therefore, w.l.o.g. we can assume that $U'_i\subset S_i$, $1\leq i\leq 3$. Now, for each $j\in\{1,2,3\}$, define $I_j=\{i, v_i\in S_j\}$ and $w(I_j)=\sum_{i\in I_j} w_i$. So, $\{I_1,I_2,I_3\}$ is a partition of $\{1,2,\ldots,n+1\}$ and for each $j\in\{1,2,3\}$, we have
\begin{align*}
\phi(S_j)&= \frac{2(w(I_j)+N|I_j|)+(k-3)N/3 +3B+(n+1)N- w(I_j)-N|I_j|}{|I_j|+M}\\
& = \frac{w(I_j)+N|I_j|+3B+(n+1)N+(k-3)N/3 }{|I_j|+M}\\
& = \frac{w(I_j)+N|I_j|+NM-B}{|I_j|+M}\leq N.
\end{align*}
Thus, for each $j\in\{1,2,3\}$, $w(I_j)\leq B$. Now, since $\sum_{i=1}^{n+1} w_i=3B$, we have $w(I_j)=B$, $1\leq j\leq 3$. Also, w.l.o.g. we may assume that $n+1\in I_3$ and therefore, $I_1\subset \{1,2,\ldots, n\}$ and $w(I_1)=B$. This completes the proof.
\end{proof}
\vspace{-0.4cm}
\begin{figure}[H]
  \centering
  \begin{tikzpicture}[scale = 0.9]
 
  \draw[draw=black] (0,0) rectangle (0.75,5);
  \node at (-0.5,2.5) {$X$} ;
  \node at (8.2,2.5) {$Y$} ;

  \node[inner sep=1pt, fill=black,circle,draw] (u1) at ({0.37}, {3.5}) {};
  \node[inner sep=1pt, fill=black,circle,draw] (u2) at ({0.37}, {2.5}) {};
  \node[inner sep=1pt, fill=black,circle,draw] (u3) at ({0.37}, {1.5}) {};
  
   \node at (0.37,3.8) {$u_{1}$} ;
   \node at (0.37,2.8) {$u_{2}$} ;
   \node at (0.37,1.8) {$u_{3}$} ;

 \draw[draw=black] (2.7,-0.5) rectangle (7.7,5);
    
 \draw[draw=black, rounded corners, thick] (3,4) rectangle (7,4.5);
 \draw[draw=black, rounded corners, thick] (3,2.5) rectangle (7,3);
 \draw[draw=black, rounded corners, thick] (3,1) rectangle (7,1.5);
 \draw[draw=black, rounded corners, thick] (3,0.5) rectangle (7,1);
  \draw[draw=black, rounded corners, thick] (3,0) rectangle (7,0.5);

   \node[inner sep=1pt, fill=black,circle,draw] (v1) at ({3.75}, {4.25}) {};
   \node[inner sep=1pt, fill=black,circle,draw] (v2) at ({4.75}, {4.25}) {};
  \node at (5.6,4.25) {$\textbf{\dots}$} ;
   \node[inner sep=1pt, fill=black,circle,draw] (vn+1) at ({6.3}, {4.25}) {};
  \node at (7.4,4.25) {$V$} ;
  \node[inner sep=1pt, fill=black,circle,draw] (y1) at ({3.75}, {2.75}) {};
   \node[inner sep=1pt, fill=black,circle,draw] (y2) at ({4.75}, {2.75}) {};
  \node at (5.6,2.75) {$\textbf{\dots}$} ;
   \node[inner sep=1pt, fill=black,circle,draw] (yn+1) at ({6.3}, {2.75}) {};
  \node at (7.4,2.75) {$\hat{Y}$} ;

\node at (5.1,0.25) {$u_{1}^{1} \dots u_{M-1}^{1}$} ;
\node at (5.1,0.75) {$u_{1}^{2} \dots u_{M-1}^{2}$} ;
\node at (5.1,1.25) {$u_{1}^{3} \dots u_{M-1}^{3}$} ;
\node at (7.4,0.75) {$U_{2}$} ;
\node at (7.4,0.25) {$U_{1}$} ;
\node at (7.4,1.25) {$U_{3}$} ;

  
\draw[] (u1) to ({3}, {4.25});
\draw[] (u1) to ({3}, {2.75});
\draw[] (u1) to ({3}, {0.25});

\draw[] (u2) to ({3}, {4.25});
\draw[] (u2) to ({3}, {2.75});
\draw[] (u2) to ({3}, {0.75});

\draw[] (u3) to ({3}, {4.25});
\draw[] (u3) to ({3}, {2.75});
\draw[] (u3) to ({3}, {1.25});

  \end{tikzpicture}
  \caption{Schema of the adjacencies in the bipartite graph $G = (X, Y)$.}
  \label{fig:cbg}
\end{figure}

\begin{corollary}\label[cor]{cor:NPtw3}
For every fixed integer $k\geq 3$ and $t\geq 3$, the problem $k$SC is NP-hard for the graphs with treewidth at most $t$.
\end{corollary}
%
For our next result, we show that \textsc{$k$-Sparsest Cut} remains NP-hard on graphs with maximum degree at most three and also on graphs with degeneracy at most two. The idea is similar to the one in \cite{feige1997densest}. 

\begin{theorem}\label{thm:2}
For every fixed integer $k\geq 2$,
\begin{itemize} 
\item[\rm (i)] the problem $k$SC is NP-hard for the graphs with maximum degree three, and
\item[\rm (ii)] the problem $k$SC is NP-hard for the graphs with degeneracy two.
\end{itemize}	
\end{theorem}

\begin{proof}
	We give a reduction from $k$SC for general (unweighted) graphs which is known to be NP-hard for every fixed integer $k\geq 2$ \cite{daneshgar2012complexity}. 
	Let $G$ be an instance of $k$-\textsc{Sparsest Cut}, where $G = (V, E)$ is a simple unweighted graph, where $ V(G)= \{v_{1}, v_{2}, \dots, v_{n} \}$.
	%
	We construct a weighted graph $G^{\prime}$ with maximum degree three and a weight function $w:E(G')\to \mathbb{R}^+$ as follows.
	
	For every $1 \leq i \leq n$, let $C^{i}$ be an $n$-cycle on vertices ($v_{1}^{i}, \dots, v_{n}^{i}$) and let $V(G')=\cup_{i=1}^n V(C^i)$. For each edge of $G$, say $e = \{ v_{i}, v_{j} \}$, create an edge $e^{ij}$ in $G'$ between vertices $v_{j}^{i}$ and $v_{i}^{j}$ and let $w(e^{ij}) = 1$. Also, let the weights of the edges in the cycles $C^i$ be a sufficiently large integer $M$.
	It is clear that the construction can be done in polynomial time and the obtained graph $G'$ has a maximum degree of three. Now, for every $k$-partition $\{S_1,\ldots,S_k\}$ of $V(G)$ and every $j\in [k]$, let $S'_j=\cup_{v_i\in S_j}{V(C^i)}$.
	Therefore, $\phi(S'_j)= ({1}/{n}) \phi(S_j)$. Moreover, since the edge weights of cycles, $C^i$ are large enough, in every minimizer for $G'$, all vertices of each $C^i$ appear in the same part. Hence, $\phi_k(G)=n \phi_k(G')$. So, the reduction preserves the edge expansion and this proves (i).
	In order to prove (ii), first, note that the $k$-\textsc{Sparsest Cut} problem remains NP-hard for regular graphs \cite{raghavendra2010graph}. So, assume that $G$ is a $d$-regular graph. For each edge, $e=\{v_i,v_j\}$ of $G$, replace the edge $e^{ij}$ in $G'$ with a path of length three $P_e=acdb$ such that $w(cd)=1$ and $w(ac)$ and $w(db)$ are equal to a sufficiently large  integer (the edges in each cycle $C_i$ remain unchanged). Call the obtained graph $G''$ and it is clear that the degeneracy of $G''$ is equal to two (since the vertices of degree three induce a stable set). Now, let $\{S'_1,\ldots,S'_k\}$ be as above. Let $e=uv$ be an edge between $S'_i$ and $S'_j$, where $u\in S'_i$ and $v\in S'_j$. Also, let $P_e=uxyv$ be its corresponding path in $G''$. Then, we add $x$ to $S'_i$ and $y$ to $S'_j$. We do this for all such edges to obtain a $k$-partition $\{S''_1,\ldots, S''_k\}$ for $V(G'')$. Let $\phi(S_i)={m_i}/{|S_i|}$, where $m_i$ is the number of outgoing edges from $S_i$. Then, it is clear that $\phi(S''_i)={m_i}/(n|S_i|+d|S_i|)$. Therefore, $\phi_k(G)\leq N$ if and only if $\phi_k(G'')\leq N/(n+d)$. This completes the proof. 
\end{proof}

The following can be concluded immediately.

\begin{corollary}\label[cor]{cor:delta}
	The problem $k$SSE is NP-hard for the graphs with maximum degree three and also for the graphs with degeneracy two.
 \end{corollary}

\begin{corollary}\label[cor]{cor:k,delta}
 The problem $k$SC is W[1]-hard for $(k,\Delta)$ combined and also $(k,d)$ combined, where $\Delta$ and $d$ are respectively the maximum degrees and the degeneracy of the input graph.
\end{corollary}

\section{Unweighted version}\label{sec4}
In this section, we consider the unweighted version of the $k$-\textsc{Sparsest Cut} problem, that is, when the edge weights are equal to one. First, we show that the problem is W[1]-hard parameterized by three combined parameters tree-depth, feedback vertex set number, and $k$. This also implies the hardness of the problem with respect to the treewidth of the input graph. Moreover, we show that no algorithm can solve $k$-\textsc{Sparsest Cut} in time $n^{o(k+\tw/\log(k+\tw))}$, unless ETH is false.
\begin{theorem}\label{thm:w1}
	The unweighted $k$SC problem is W[1]-hard when parameterized by the three combined parameters tree-depth, feedback vertex set number, and $k$.
\end{theorem}
\begin{proof}
	We give a parameterized reduction from \textsc{Unary Bin Packing} parameterized by the number of bins defined as follows.
\decprob{\textsc{Unary Bin Packing}}{Positive integers $w_{1},\dots, w_{l}, b, C$ each encoded in unary.}{Can we partition $l$ items with weights $w_{1},\dots,w_{l}$ into $b$ bins such that sum of the weights in each bin does not exceed $C$?}

	Jansen et al. \cite{jansen} showed that \textsc{Unary Bin Packing} is W[1]-hard when parameterized by the number $b$ of bins, and more precisely, that the problem cannot be solved in $f(b) \cdot n^{o(b/\log b)} $ time for any computable function $f$ for an input of length $n$ under the ETH. Let us consider an instance of \textsc{Unary Bin Packing} as $I=(w_{1},\dots,w_{l},b,C)$. Also, let $W=\sum_{i=1}^l w_i$. If $W> b \cdot C$, then evidently it is a NO-instance. Without loss of generality, we may assume that $W=b \cdot C$, since otherwise, we can add $b \cdot C-W$ items of weights equal to one. Then, we construct an instance $I'$ for $k$SC. For our convenience, first, we construct a weighted instance of $k$SC in which vertices are weighted. Then, using a unitarization technique given in \cite{daneshgar2012complexity}, we construct an unweighted instance of $k$SC. When $w:V(G)\to \mathbb{R}^+$ is a vertex weight function, the edge expansion of $S\subseteq V(G)$ is defined as $\phi(S)={|E(S,\overline{S})|}/{w(S)}$. The instance $I'$ consists of a weighted complete bipartite graph $G$ defined as: $V(G):=\{u_{1},u_{2},\dots,u_{b},v_{1},\dots,v_{l+b}\}, \hspace{2mm} E(G):= \{ u_{j}v_{i}, 1\leq j \leq b, 1\leq i\leq l+b \},$
$w(u_{j})= M-\epsilon, w(v_{i}):=w_{i}+C+1+B, \forall\ 1\leq j\leq b, 1\leq i \leq l,$ and,
$w(v_i):=1+B, \forall\ l+1\leq i\leq l+b$, where $M={(l+b)(C+1+B-\epsilon)}/{(b-2)}$,  $\epsilon$ is an arbitrary small number such that $0<\epsilon \leq 1/(l+b)$ and  $B$ is a constant integer that will be determined later. Also, let all edge weights be equal to one. Let $k = b$ and $X = (l+b)/M$. So, we have an instance $ I'=(G,k,X) $ of $k$SC.  
	
	First, suppose that $I$ is a YES-instance for \textsc{Unary Bin Packing}. Then, there is a partition of $\{1,\ldots, l\}$ into $b$ bins $A_1,\ldots, A_b$ such that $\sum_{i\in A_j}{w_i} = C$, for each $j\in\{1,\ldots, b\}$. Now, for each $j\in\{1,\ldots, b\}$, define $S_j:= \{v_i: i\in A_j\}\cup \{u_j, v_{l+j}\}$. Therefore,
	
	\begin{align*}
	\phi(S_j)&= \frac{l+b-|A_j|-1+(|A_j|+1)(b-1)}{M-\epsilon+(C+1+B)|A_j|+C+1+B} \\ &=\frac{l+b+(b-2)(|A_j|+1)}{M-\epsilon+(C+1+B)(|A_j|+1)} \\
	&\leq \frac{l+b+(b-2)(|A_j|+1)}{M+(C+1+B-\epsilon)(|A_j|+1)} \\
	&\leq \frac{l+b+(b-2)(|A_j|+1)}{M+\frac{(b-2)M}{l+b}(|A_j|+1)}\\
	&=\frac{l+b}{M}=X.
	\end{align*}
	Hence, $I'$ is a YES-instance for $ kSC $. 
	
	Now, for the converse suppose that $I'$ is a YES-instance for $k$SC and $S_1,\ldots, S_b$ is a partition of $V(G)$ such that $\phi(S_j)\leq X=(l+b)/M$. First, we want to prove that each $S_j$ contains exactly one $u_j$. For if $S_j$ contains none of vertices $u_1,\ldots,  u_b$, then let $W_0=\max w_i$ and we have
	
	\[\phi(S_j)\geq \frac{|S_j|b}{(W_0+C+B+1)|S_j|}= \frac{b}{W_0+C+B+1}.  \]
	Now, if we choose integer $B$ such that $2B> W_0(b-2)+b\epsilon$, then 
	\[\phi(S_j)> \frac{b-2}{C+B+1-\epsilon}= \frac{l+b}{M}=X, \]
	which is a contradiction. Therefore, each $S_j$ contains exactly one $u_j$, w.l.o.g. suppose that $u_j\in S_j$, for $1\leq j\leq b$.
	Note that if $S_j=\{u_j\}$, then $\phi(S_j)=\frac{l+b}{M-\epsilon}>X$ which is a contradiction. So, for each $j\in\{1,\ldots b\}$, we have $S_j=A_j\cup \{u_j\}$, where $A_j$ is a non-empty subset of $\{v_1,\ldots, v_{l+b}\}$. 
	
	Therefore, for each $j\in\{1,\ldots ,b\}$
	
	\begin{align*}
	\frac{l+b}{M}\geq \phi(S_j)= \frac{l+b+(b-2)|A_j|}{M-\epsilon+w(A_j)}.
	\end{align*}
	
	Thus, 
	$$C+B+1-\epsilon=\frac{M(b-2)}{l+b}\leq \frac{w(A_j)-\epsilon}{|A_j|}, $$
	and we have 
	\[w(A_j)-(C+B+1)|A_j|\geq -\epsilon(|A_j|-1)>-\epsilon (l+b)\geq -1. \]
	Since the value of the left-hand side is an integer, then 
	
	\begin{equation}\label{eq:Aj}
	w(A_j)\geq (C+B+1)|A_j|.
	\end{equation}
	
	Hence, since $\sum_{j=1}^b	w(A_j)= (C+B+1) (l+b)$, we have $w(A_j)= (C+B+1)|A_j|$, for all $j\in \{1,\ldots,b\}$.
	If one $A_j$ contains none of vertices in $\{v_{l+1},\ldots,v_{l+b}\}$, then
	$w(A_j)> (C+B+1)|A_j|$ which is a contradiction. Thus, each $A_j$ contains exactly one of vertices in $\{v_{l+1},\ldots,v_{l+b} \}$. Let us set $B_j= \{i: 1\leq i\leq l, v_i\in A_j\}$. Then, for each $j\in\{1,\ldots,b\}$, 
	
	\[\sum_{i\in B_j}w_i= w(A_j)-(B+C+1)|A_j|+C= C. \]
	
This shows that $I$ is a YES-instance for \textsc{Unary Bin Packing}.	
	Finally, we show how we can construct a graph $G'$ from the graph $G$ where the vertices in $G'$ are unweighted. For this, we use the unitarization technique given in \cite{daneshgar2012complexity}. First, we choose a large integer $\chi$ such that for every vertex $u$, $\chi w(u)\geq |E(G)|$. Then, for every vertex $u$, we add a set of exactly $\chi w(u)-1$ new vertices and join all of these vertices to $u$. In the obtained graph $G'$, all edge and vertex weights are equal to one and it is easy to see that $\phi_k(G)=\chi \phi_k(G')$ (for a precise proof, see \cite{daneshgar2012complexity}). 
	Also, since the weights are written in unary code, it is a polynomial-time process. 
	
	Finally, note that since $G$ is a complete bipartite graph,
	both feedback vertex set number and tree-depth of $G$ are bounded by $b$, and since adding pendant vertices does not change the feedback vertex set number and increase the tree-depth by at most one, we see that both parameters are still bounded by $b$ after the unitarization process. 
	
	Now since $k=b$, this is a parameterized reduction w.r.t feedback vertex set number, tree-depth, and the number $k$, combined.
\end{proof}

The above proof also implies that the $k$SC remains W[1]-hard w.r.t the treewidth of the input graph and the number $k$, combined.
We state this observation as the following corollary of \cref{thm:w1} and leave the reader to check the detail on the algorithmic barrier of graph-width parameters \cite{belmonte2020grundy}.

\begin{corollary}\label[cor]{cor:tw}
The unweighted $k$SC problem is W[1]-hard when parameterized by the treewidth of the input graph and the number $k$ combined.
\end{corollary}

\begin{theorem}\label{eth}
Assuming ETH, the unweighted $k$SC cannot be solved in time $f(\alpha)n^{o(\alpha/\log \alpha)}$, for a computable function $f$, where $\alpha=\max\{k,\tw\}$ and $\tw,n$ are respectively the treewidth and the order of the input graph.
\end{theorem}

\begin{proof}
Let $I = (w_{1}, \dots, w_{l}, b, C)$ be an instance of length $n$ for \textsc{Unary Bin Packing}, i.e. $n=O(w_1+\cdots+w_l+b+C)$. \cref{thm:w1} shows that we can construct an instance $I^{\prime} = (G, k, X)$ of $k$SC in $f_{1}(b)n^{O(1)}$ time such that $k = \tw(G)=b$. Now, if we could solve $k$SC in $f_{2}(\alpha)|I^{\prime}|^{o(\alpha/ \log(\alpha))}$ time for some function $f_{2}$, then we would be able to solve \textsc{Unary Bin Packing} in $f_{3}(b)n^{o(b/\log b)}$ for some function $f_{3}$. This would imply that ETH fails by \cite{jansen}.
\end{proof}

For our next result, we present an FPT algorithm for $k$-\textsc{Sparsest Cut} when the parameter is the treewidth of the input graph and $k$ is a fixed integer.
Daneshgar et al. \cite{daneshgar2012complexity} proved that when $k$ is in input, the problem $kSC$ is NP-complete even for unweighted trees. 
Here, we prove that when $k$ is a fixed integer, the problem $kSC$ admits an FPT algorithm with respect to the treewidth of the input graph.  

\begin{theorem}\label{thm:tw}
For every fixed integer $k \geq 2$, the unweighted $k$SC can be solved in $ O(k^{\tw} n^{6k+1})$ and uses space exponential to $\tw$, where $n$ and $\tw$ are respectively the order and the treewidth of the input graph.
\end{theorem}

\begin{proof}
The idea is similar to the one in the proof of \Cref{thm:twk=2} and is based on a dynamic program that computes the values of a table on the nodes of the tree decomposition of the graph in a bottom-up fashion. For convenience and easier analysis, we use a nice tree decomposition.

Let $\tw$ be the treewidth of $G$ and consider a nice tree decomposition $ (T, \{ X_{t} \}_{t\in V(T )} ) $ for $G$ of width $\tw$ as defined in Section~\ref{sec:prelim}. Let $T_{i}$ be the subtree of $T$ rooted at node $i\in V(T)$ and $G_{i} = (V_{i}, E_{i})$ be the subgraph of $G$ induced by the vertices in $\bigcup_{j \in V(T_{i})} X_{j}$. 
For each node $i\in V(T)$, consider a solution $(S_1,\ldots, S_k)$ to $k$SC on $G_i$ as a $k$-partition of $V_i$ and define a configuration vector $c \in \{1,\ldots,k\}^{|X_{i}|}$, where $c[u] = j$, iff  $u\in X_i\cap S_j$. 

 Also, define two vectors $\nu\in \{0,\ldots, n\}^k$ and $\mu\in \{0,\ldots, n^2\}^k$, where $\nu[j]$ equals to the size of $S_j$ and $\mu[j]$ is equal to the size of $E(S_j,\overline{S_j})$. Moreover, consider a table $A$ with $|V(T)|$ rows and at most $k^{\tw}n^{3k}$ columns, where each row of $A$ represents a vertex $i\in V(T)$ and each column of $A$ represents a configuration vector $ c $ and two vectors $\nu$ and $\mu$.
The value of an entry of this table $ A[i,c,\nu,\mu] $ is equal to one if and only if there is a $k$-partition $(S_1,\ldots, S_k)$ of $V(G_i)$ such that for each $j\in\{1,\ldots, k\}$, $S_j\cap X_i=\{u\in X_i \: \ c[u]=j\}$, $|S_j|=\nu[j]$ and $|E(S_j,\overline{S_j})|\leq \mu[j]$. If there is no such partition, we define $A[i,c,\nu,\mu]=0$.

\subparagraph*{Leaf nodes.}
The algorithm examines the nodes of $T$ in a bottom-up manner and fills in each row of the table $A$. In the initialization step, for each leaf node $i\in V(T)$, since $X_i$ and $V_i$ are both empty, we have $A[i,c,\nu,\mu] =1$ if and only if for all $j\in\{1,\ldots, k\}$, $\nu[j]=0$.  
\subparagraph*{Forget nodes.}
Let $i$ be a forget node with the child $X_l$, where $X_i=X_l\setminus \{v\}$. Then, consider a configuration $c$ and two vectors $\nu,\mu$ for node $i$. 
The configuration $c$ is extended to a one for node $l$ by deciding the part to which vertex $v$ belongs. So, for every $j\in\{1,\ldots, k\}$, define configuration $c_j$, where $c_j[u]=c[u]$, for every $u\in X_i$ and $c_j[v]=j$.   Hence, we get
$$ A[i,c,\mu,\nu] = \max_{1\leq j\leq k} A[l,c_j,\mu,\nu]. $$

\subparagraph*{Introduce nodes.}
Let $i$ be an introduce node with child $l$, where $X_i=X_l\cup \{v\}$ and $v\notin X_l$. Now, consider a configuration $c$ and two vectors $\mu,\nu$ for node $i$. Suppose that $c[v]=j_0$ and let $c'$ be the configuration obtained by restriction of $c$ to $X_l$. Also, define the vectors $\mu'$ and $\nu'$ as follows.

\begin{align*}
\nu'[j]&=\begin{cases}
\nu[j] & \text{ if }j\neq j_0, \\
\nu[j]-1 & \text{ if }j =j_0,
\end{cases} \\
\mu'[j]&= \begin{cases}
\mu[j]-|E(c^{-1}(j),\{v\})| & \text{ if }j\neq j_0, \\
\mu[j]-|E(\{v\},X_l\setminus c^{-1}(j_0))| & \text{ if }j= j_0.
\end{cases}
\end{align*}

Therefore,  
\[A[i,c,\mu,\nu]= A[l,c',\mu',\nu'].\]

To see this, note that since $v$ is in $S_{j_0}$, if we remove vertex $v$, the size of $S_j$ is subtracted by one if $j=j_0$ and does not change, otherwise. Also, since all  neighbors of $v$ are in $X_l$, $|E(S_{j_0},\overline{S_{j_0}})|$ is reduced by $|E(\{v\}, \overline{S_{j_0}})|=|E(\{v\}, X_l\setminus c^{-1}(j_0))|$ and for $j\neq j_0$,  $|E(S_{j},\overline{S_{j}})|$ is reduced by $|E(S_j,\{v\})|=|E(c^{-1}(j),\{v\})|$.
  
\subparagraph*{Join nodes.}
Let $i$ be a join node with two childs $l_1,l_2$, where $X_i=X_{l_1}=X_{l_2}$. Also, consider a configuration $c$ and two vectors $\mu,\nu$ for node $i$. 
Therefore, 

\[A[i,c,\mu,\nu]= \max_{\stackrel{\mu_1,\mu_2}{\nu_1,\nu_2}} \min_{1\leq t\leq 2} A[l_t,c,\mu_t,\nu_t],\]
where the maximum is taken over all vectors $\nu_1,\nu_2,\mu_1,\mu_2$ such that for every $j\in\{1,\ldots, k\}$, 
\begin{align*}
\nu[j]&= \nu_1[j]+\nu_2[j]- |c^{-1}(j)|, \\
\mu[j]&= \mu_1[j]+\mu_2[j]- |E(c^{-1}(j), X_i\setminus c^{-1}(j))|.
\end{align*}

To see this, note that $V_{l_1}\cap V_{l_2}=X_i$. So, each $k$-partition of $V_i$ can be divided into two $k$-partitions for $V_{l_1}$ and $V_{l_2}$ where they both agree on $X_i$. For a subset $S\subset V_i$, if we define $S_t=S\cap V_{l_t}$, $1\leq t\leq 2$, then $|S|=|S_1|+|S_2|-|S_1\cap S_2|=|S_1|+|S_2|-|S\cap X_i|$. Also, 
$|E(S,V_i\setminus S)|=|E(S_1,V_{l_1}\setminus S_1)|+|E(S_2,V_{l_2}\setminus S_2)|- |E(S\cap X_i, X_i\setminus S)|$.

Each entry of the table for leaf, forget, introduce and join nodes can be computed in worst cases $O(1), O(k), O(k), O(n^{3k}) $, respectively. Also, the size of the table is $O(k^{\tw} \times n^{3k+1}) $. Hence, the runtime of the whole algorithm is at most $O(k^{\tw} n^{6k+1})$.
\end{proof}

\begin{remark}
By W[1]-hardness result of Theorem~\ref{thm:w1},  the running time of the algorithm is unlikely to be improved to $f(k,\tw)n^{O(1)}$.
\end{remark}

\Cref{thm:tw} implies that for every fixed $k\geq 2$ the unweighted $k$SC is FPT with the parameter treewidth $\tw$, while, in contrast, the weighted version is NP-hard for every $k,\tw\geq 3$ (\Cref{cor:NPtw3}).
Also, since treewidth of a graph is bounded by its vertex cover number, this implies fixed-parameter tractability of the unweighted $k$SC with the parameter vertex cover number $\vc$, in the sense that $k$SC can be solved in time $O(k^{\vc})\, n^{O(k)}$. This gives rise to the question that if $k$SC can be solved in time $f(k,\vc)\, n^{O(1)}$ for some computable function $f$. In the following theorem, we give an affirmative answer to this question by proving that $k$SC is FPT with respect to $(k,\vc)$. 

\begin{theorem}\label{thm:vc}
	For every fixed integer $k\geq 2$,  there is an algorithm that solves the unweighted $k$SC in time $2^{O(k^2\vc^k \log \vc)}\, n^{O(1)}$ and in polynomial space, where $n$ and $\vc$ are respectively the order and the size of the vertex cover number of the input graph.
\end{theorem}

In order to prove this theorem, we need the following classical result by Lenstra which states that \textsc{Integer Programming} is FPT with respect to the number of variables. 

\begin{theorem} {\rm \cite{lenstra}} \label{thm:ilp}
An \textsc{Integer Linear Programming} instance of size $L$ with $p$ variables can be solved using $O(p^{2.5p+o(p)}\, L)$ arithmetic operations and space polynomial in $L$.		
\end{theorem}

\begin{proof}[Proof of \Cref{thm:vc}.]
Suppose that the graph $G$, integer $k\geq 2$ and rational number $N$ is given and we are going to solve the decision problem if $\phi_k(G)\leq N$?	
	
Let $C$ be a vertex cover of the graph $G$ with size $\vc$ and define $I=V(G)\setminus C$. 
The algorithm is based on seeking over all $k$-partitions $(C_1,\ldots, C_k)$ of $C$ and solving an integer program with a bounded number of variables for each such partition.

	Let $C_1,\ldots, C_k$ be a fixed partition of $C$. Also, let $\mathcal{A}$ be the set of all vectors $\mathbf{a}=(a_1,a_2,\ldots, a_k)$ with integer components such that $0\leq a_i\leq |C_i|$ for each $i\in [k]$. Now, for each $\mathbf{a}\in \mathcal{A}$, define $I_{\mathbf{a}}$ to be the set of all vertices in $I$ which has exactly $a_i$ neighbors in $C_i$, for each $i\in [k]$ and let $n_{\mathbf{a}}=|I_{\mathbf{a}}|$.
	
	Now, we distribute vertices of $I$ into the sets $C_1,\ldots, C_k$ and for each vector $\mathbf{a}\in \mathcal{A}$ and integer $i\in [k]$, we define a variable $x_{\mathbf{a},i}$ which indicates the number of vertices in $I_{\mathbf{a}}$ which are added to the set $C_i$. Let $S_i$ be the set obtained from $C_i$ by adding exactly $x_{\mathbf{a},i}$ vertices from $I_{\mathbf{a}}$ for each $\mathbf{a}\in \mathcal{A}$. Then, $(S_1,\ldots, S_k)$ is a $k$-partition of $V(G)$ and for each $i\in [k]$, we have

\begin{align*}
\phi_G(S_i)& = \dfrac{|E(C_i,V(G)\setminus C_i)|+\displaystyle \sum_{\mathbf{a}\in \mathcal{A}} (a_1+\cdots+a_k-2a_i)\, x_{\mathbf{a},i} }{|C_i|+\displaystyle\sum_{\mathbf{a}\in \mathcal{A}} x_{\mathbf{a},i}}.
\end{align*}

Now, for every $i\in [k]$, define $W_i=|E(C_i,V(G)\setminus C_i)|$ and $w_i=a_1+\cdots+a_k-2a_i$. Therefore, we have $\phi_G(S_i)\leq N$ if and only if $W_i+\sum_{\mathbf{a}\in \mathcal{A}} w_i\, x_{\mathbf{a},i}\leq N(|C_j|+\sum_{\mathbf{a}\in \mathcal{A}} x_{\mathbf{a},i})$.

Hence, there exists a $k$-partition $(S_1,\ldots, S_k)$ of $V(G)$ such that for each $i\in [k]$, $S_i\cap C=C_i$ and $\phi_G(S_i)\leq N$ if and only if the following integer programming has a feasible solution.

\begin{align*}
&\sum_{i=1}^k x_{\mathbf{a},i} = n_{\mathbf{a}}, &\forall\, \mathbf{a}\in \mathcal{A},\\
&\sum_{\mathbf{a}\in \mathcal{A}} (w_i-N)\, x_{\mathbf{a},i}\leq N|C_j|-W_i, &\forall\, i\in [k].
\end{align*}

Since $|\mathcal{A}|\leq \vc^k$, the number of variables is at most $k\vc^k$ and by \Cref{thm:ilp}, this integer programming can be solved in time $ (k\vc^k)^{O(k\vc^k)}\, n^{O(1)}$.
Also, note that the number of $k$-partitions $(C_1,\ldots, C_k)$ of $C$ is at most $k^{\vc}$.
Therefore, the running time of the whole algorithm is at most $ 2^{O(k^2\vc^k \log \vc)}\, n^{O(1)}$. Also, each ILP can be solved separately in a polynomial space.
\end{proof}


\section{Small-Set Expansion in Degree-Bounded Graphs}\label{sec5}

In \Cref{cor:delta}, we showed that $k$SSE is NP-hard for graphs with maximum degree three and so is W[1]-hard for the parameter $\Delta$. In this section, we show that $k$SSE is also W[1]-hard for $k$. Also, using a random separation technique, we will prove that $k$SSE is FPT with respect to $(k,\Delta)$. 

\begin{theorem}\label{thm:kssew1k}
The unweighted $k$SSE is W[1]-hard for the parameter $k$. 
\end{theorem}
\begin{proof}
We give a parametrized reduction from $k$-clique on regular graphs which is known to be W[1]-hard concerning $k$ \cite{cai2008parameterized, marx2008parameterized}. Let $G$ be a $d$-regular graph. First, if $G$ has a $k$-clique, say $S$, then we have 
\[\phi(S)=\frac{dk-k(k-1)}{k}=d-k+1, \]
and thus $\psi_k(G)\leq d-k+1$.
Now, suppose that $\psi_k(G)\leq d-k+1$. Then, there is a subset $S\subseteq V$, where $|S|\leq k$ and $\phi(S)\leq d-k+1$. Let $s$ be the size of $S$. Now, we have
\[d-k+1\geq \phi(S)=\frac{|E(S,\overline{S})|}{|S|}\geq \frac{ds-s(s-1)}{s}=d-s+1. \]
Therefore, $s\geq k$. This implies that $|S|=k$. Also, if $S$ is not a clique, then the last inequality is strict which is a contradiction. Hence, $S$ is a $k$-clique. This completes the proof.
\end{proof}

In the sequel, we will show that $k$SSE is FPT with respect to $(k,\Delta)$ combined, while it is W[1]-hard for each of the parameters separately. Our algorithm relies on   
random separation technique \cite{cai2006random}.
\begin{theorem}\label{thm:randomsep}
The problem \textsc{$k$SSE} is fixed-parameter tractable when parameterized by $k$ and $d$, where $d$ is the maximum degree of the input graph.
\end{theorem}

\begin{proof}
Let $G=(V,E)$ be a graph with maximum degree $d$ endowed with a weight function $w:E\to \mathbb{R}^+$. We randomly color each vertex of $G$ by either green or red (with equal probabilities) to see a random partition $(V_{g}, V_{r})$ of $V$. Let $\hat{S}\subset V$ be a solution to $k$SSE. A partition $(V_g,V_r)$ is called \textit{good}  for $\hat{S}$ if $\hat{S}\subset V_g$ and  $N_{G}(\hat{S})\setminus \hat{S}\subset V_r$, i.e. all vertices in $\hat{S}$ are green and also, all vertices in $V\setminus \hat{S}$ with a neighbor in $\hat{S}$ are red. We can see that the probability that a random partition is a good partition for $\hat{S}$ is at least $2^{-(d+1)k}$. Note that in a good partition, $\hat{S}$ is the union of some green connected components since each green connected component must either be contained in $\hat{S}$ or completely excluded from $\hat{S}$.
 
Now, fix a good partition $(V_g,V_r)$ and let $C_1,\ldots, C_t$ be the connected components of $G$ induced on $V_g$. Also, for each $i\in \{1,\ldots ,t\}$, let $n_i=|C_i|$ and $m_i= w(C_i,V\setminus C_i)$. Now, finding a solution $S\subseteq V_g $ with $|S|\leq k$ and $\phi_G(S)\leq N$, is reduced to the problem of finding a subset $I\subseteq \{1,\ldots, t \}$, such that, 
$ \frac{\sum_{i\in I} m_{i}}{\sum_{i\in I} n_{i}} \leq N$, and $\sum_{i\in I}n_i\leq k.$
This problem can be solved in $O(kn)$ time using the standard dynamic programming algorithm for 0-1 knapsack problem \cite{kleinberg2006algorithm}. Moreover, the computation of the green components and the numbers $n_{i}$ and $m_{i}$ can be done in $O(dn)$ time. Thus, we can find a solution $\hat{S}$ in $O((d+k)n)$ time with probability at least $2^{-(d+1)k}$.

For derandomization, we use the notion of \textit{universal set}. Let $A \subseteq \{0, 1\}^{n}$ contain binary strings of length $n$. We say $A$ is an $(n, k)$-\textit{universal} if for every array $I = (i_{1}, i_{2}, \dots, i_{k})$ of $k$ string positions, the projection $A|_{I} = \{(a_{i_{1}}, a_{i_{2}}, \dots, a_{i_{k}}) \hspace{1mm} | \hspace{1mm} a = (a_{1}, a_{2}, \dots, a_{n}) \in A \} $, 
contains all $2^{k}$ possible binary strings of length $k$.

Naor et al. \cite{naor1995splitters} present a near-optimal deterministic construction of a $(n,k)$-universal set of size $2^{k}k^{O(\log k)}\log n$. Now, we choose a $(n,(d+1)k)$-universal set of this size as a collection of partitions of $V$. For each solution $\hat{S}$ to $k$SSE, we have $|\hat{S}|\leq k$ and $|N(\hat{S})|\leq dk$. Therefore, there is a partition in the universal set which is a good partition for $\hat{S}$. So, running the above algorithm for each of the partitions in the universal set can find the solution. 

The used universal set contains at most $2^{(d+1)k}(dk+k)^{O(\log (dk+k))}\log n$ partitions. So the running time of our deterministic algorithm is at most $ O(f(k,d)n \log n)$, where $f(k,d) = 2^{(d+1)k}(dk+k)^{O(\log (dk+k))}(d+k)$.
\end{proof}

\section{Conclusion}
We have presented exact and parameterized algorithms as well as some hardness results for the $k$-\textsc{Sparsest Cut} and the $k$-\textsc{Small-Set Expansion} problems. Our algorithms deal with many parameters such as treewidth, vertex cover number, maximum degree, and degeneracy of the input graph. It would be challenging to improve the running time of the presented algorithms.
In addition, investigating parameterized complexity of unweighted $k$SC with respect to other structural parameters such as clique-width, modular-width and neighborhood diversity can be the subject of future work and can shed more light on the ``price of generality'' line of research.

	\bibliographystyle{abbrv}
    \bibliography{reference}
\end{document}